\documentclass[12pt]{article}

\usepackage[T1]{fontenc}
\usepackage{kpfonts}
\usepackage{eurosym}
\usepackage{graphicx}
\usepackage{enumerate}
\usepackage{amsmath}
\usepackage{amsfonts}
\usepackage{amssymb}
\usepackage{amsfonts}
\usepackage{amssymb,amsmath,amsthm,epsfig,graphicx,natbib,subfigure,hyperref}
\usepackage{geometry}
\usepackage{caption}
\usepackage{color}
\usepackage{enumerate}
\usepackage{setspace}
\usepackage[textwidth=30mm]{todonotes}
\usepackage{sgame, tikz}
\usepackage{multirow,array}
\usepackage{rotating}
\usepackage{graphicx}
\usepackage{subfigure}
\usepackage{float}
\usepackage{booktabs}
\usepackage{makecell}
\usepackage{xurl}
\usepackage{threeparttable}
\usepackage{subcaption}
\usepackage{hyperref}
\usepackage{pgfplots}

\hypersetup{colorlinks=true,linkcolor=blue,citecolor=blue,urlcolor=blue}
\allowdisplaybreaks

\bibpunct[: ]{(}{)}{;}{a}{}{,}

\newtheorem{theorem}{Theorem}

\newcommand {\be}{\begin{equation}}
\newcommand {\ee}{\end{equation}}

\title{A Note on Market Segmentation and Bertrand Competition\thanks{We thank Calzolari Giacomo for addressing this question.}}

\author{Zhang Xu\thanks{School of Economics, Renmin University of China, China. Email:
xuzhang@ruc.edu.cn } \and Mingsheng Zhang\thanks{School of Economics, Renmin University of China, China. Email:
mingsheng\_zhang@yeah.net } \and Wei Zhao\thanks{
School of Economics and Management, Tsinghua University, China. Email:
wei.zhao@outlook.fr }}

\usepackage{soul}

\date{\today\\[1cm] Latest Version: \href{https://365.kdocs.cn/l/cpHhDYaCcF4U}{Click Here}}
\begin{document}
\maketitle

\begin{abstract}
In this note, we show that equilibrium profit is zero in Bertrand competition with a finite number of firms and consumers whose willingness to pay are bounded, under any market segmentation profile.

\noindent \emph{Keywords:} {Personalized Pricing, Bertrand Competition, Market Segmentation}

\end{abstract}

\newpage
{
  \hypersetup{linkcolor=black}
  \tableofcontents
}

\newpage

\clearpage
\pagenumbering{arabic}
\setcounter{page}{1}

\newpage

\section{Introduction}

    The technology development in data collection, storage and processing makes price discrimination a relevant question in industrial organization. As price discrimination has been extensively studied in monopolistic setting,\footnote{One of the seminal papers is \cite{bergemann2015limits}} a natural question then arises, i.e. how price discrimination affects competition among firms. 

    Though this question has been partially answered in competition among firms with differentiated goods,\footnote{See \cite{rhodes2024aer} and \cite{Elliott2026}.} the benchmark scenario should be Bertrand competition with homogeneous goods. Can market segmentation alleviate the level of competition in Bertrand competition? This note offers an affirmative negative answer, when the number of firms is finite and consumers' willingness to pay is bounded.

\section{Model and Main Result}

There are finite numbers of firms $i\in N={1,\cdot\cdot\cdot,n}$ and a unit of consumers whose willingness to pay (WTP) $\omega \in \Omega:=[0,\overline{\omega}]$ with measure $\mu$. A market segmentation profile is given by $(\pi,\times_{i\in N} M_i)$, where $M_i=\mathbb{R}$ and a measurable function $\pi:\Omega\rightarrow \Delta(\mathbf{M}:=\times_{i\in N} M_i)$. Each firm $i$'s mixed strategy is some measurable function $\sigma_i: M_i\rightarrow \Delta(\mathbb{R}^+)$. Given price profile $\mathbf{p}:=(p_i)_{i\in N}$ and consumer with WTP $\omega$, firm $i$'s profit is given by $q_i(\mathbf{p},\omega)\ge 0$ satisfying two additional properties
\begin{equation}\label{eq:property-of-q}
    \begin{cases}
q_{i}(\mathbf{p},\omega)=0 & \text{if }p_{i}>\min\{\min_{j\in N}p_{j},\omega\}\\
\sum_{j\in N}q_{i}(\mathbf{p},\omega)=\min_{j\in N}p_{j} & \text{if }\min_{j\in N}p_{j}\le\omega
\end{cases}
\end{equation}
Given a strategy profile $\boldsymbol{\sigma}:=(\sigma_i,\sigma_{-i})$, firm $i$'s profit is given by 
\[
Q_{i}(\sigma_{i},\sigma_{-i})=\int_{\omega\in\Omega,\mathbf{m}\in M,\mathbf{p}\in[0,\infty)^{n}}q_{i}(p_{i},\mathbf{p}_{-i},\omega)\Pi_{i\in N}d\sigma_{i}(p_{i}|m_{i})d\pi(\mathbf{m}|\omega)d\mu(\omega)
\]
A strategy profile $\boldsymbol{\sigma}$ is Nash Equilibrium if 
\[
Q_{i}(\sigma_{i},\sigma_{-i})\ge Q_{i}(\hat{\sigma}_{i},\sigma_{-i}),\forall\hat{\sigma}_{i},\forall i\in N
\]

\begin{theorem}
    Among all Nash Equilibrium $\boldsymbol{\sigma}$, $Q_i \equiv 0$ for all $i\in N$.
\end{theorem}
\begin{proof}
    Fix a Nash Equilibrium $\boldsymbol{\sigma}$ and some $\mathbf{m}\in \mathbf{M}$, we have a probability measure on price profile $\mathbf{p}$, which then induces a probability measure on minimal price $\underline{p}$ as 
    \[
    d\sigma_{\underline{p}}(p|\mathbf{m})=\int_{\mathbf{p}\in[p,\infty)^{n}\backslash(p,\infty)^{n}}\Pi_{i\in N}d\sigma_{i}(p_{i}|m_{i}).
    \]
The Nash Equilibrium $\boldsymbol{\sigma}$ also induces a probability measure $G$ on transaction price 
\[
dG(p)=\int_{\omega\in[p,\overline{\omega}),\mathbf{m}\in\mathbf{M}}d\sigma_{\underline{p}}(p|\mathbf{m})d\pi(\mathbf{m}|\omega)d\mu(\omega).
\]
Denote $\overline{a} := \text{ess sup} G$. By definition, we have $\overline{a}\in[0,\overline{\omega}]$. Our objective is to prove $\overline{a} =0$.

Suppose, on the contrary, $\overline{a} >0$, then there always exists some $\hat{a}\in (\overline{a}/n,\overline{a}]$ such that $G$ is atomless in $\hat{a}$ since a sigma-finite measure has at most countable points with atom. The fact that $G$ is atomless at $\hat{a}$ induces that 
\begin{align*}
0=dG(\hat{a}) & =\int_{\omega\in[\hat{a},\bar{\omega}),\mathbf{m}\in\mathbf{M}}d\sigma_{\underline{p}}(\hat{a}|\mathbf{m})d\pi(\mathbf{m}|\omega)d\mu(\omega)\\
 & =\int_{\omega\in[\hat{a},\bar{\omega}),\mathbf{m}\in\mathbf{M}}\int_{\mathbf{p}\in[\hat{a},\infty)^{n}\backslash(\hat{a},\infty)^{n}}\Pi_{i\in N}d\sigma_{i}(p_{i}|m_{i})d\pi(\mathbf{m}|\omega)d\mu(\omega)\\
 & =\int_{\omega\in[\hat{a},\bar{\omega}),\mathbf{m}\in\mathbf{M}}\int_{\mathbf{p}\in[\hat{a},\infty)^{n}}\Pi_{i\in N}d\sigma_{i}(p_{i}|m_{i})d\pi d\mu-\int_{\omega\in[\hat{a},\bar{\omega}),\mathbf{m}\in\mathbf{M}}\int_{\mathbf{p}\in(\hat{a},\infty)^{n}}\Pi_{i\in N}d\sigma_{i}(p_{i}|m_{i})d\pi d\mu
\end{align*}
In other words,
\begin{equation}\label{eq:atomless}
\int_{\omega\in[\hat{a},\bar{\omega}),\mathbf{m}\in\mathbf{M}}\int_{\mathbf{p}\in[\hat{a},\infty)^{n}}\Pi_{i\in N}d\sigma_{i}(p_{i}|m_{i})d\pi d\mu=\int_{\omega\in[\hat{a},\bar{\omega}),\mathbf{m}\in\mathbf{M}}\int_{\mathbf{p}\in(\hat{a},\infty)^{n}}\Pi_{i\in N}d\sigma_{i}(p_{i}|m_{i})d\pi d\mu. 
\end{equation}
The total profits given that the transaction price $p\in [\hat{a},\bar{a}]$ satisfies the following inequality
\[
\int_{p\in[\hat{a},\overline{a}],\omega\in[p,\overline{\omega}],\mathbf{m}\in\mathbf{M}}pd\sigma_{\underline{p}}(p|\mathbf{m})d\pi(\mathbf{m}|\omega)d\mu(\omega)=\int_{p\in[\hat{a},\overline{a}]}pdG(p)\le\bar{a}G([\hat{a},\overline{a}]).
\]
Therefore at least one firm whose profit is at most $1/n\cdot  \bar{a}G([\hat{a},\overline{a}])$. Without loss of generality, denote one of this firm as firm $i$. Consider a new strategy $\hat{\sigma}_i$ of firm $i$ as 
\[
\hat{\sigma}_{i}(E|m_{i}):=\sigma_{i}(E\cap[0,\hat{a}))|m_{i})+\sigma_{i}([\hat{a},\infty)|m_{i})\cdot\text{dirac}_{\hat{a}}(E),\forall \text{ Borel set } E. 
\]
$i$'s profit difference can be given as 
\begin{align*}
Q_{i}(\hat{\sigma}_{i},\boldsymbol{\sigma}_{-i})-Q_{i}(\sigma_{i},\boldsymbol{\sigma}_{-i})= & \int_{\omega\in[0,\overline{\omega}],\mathbf{m}\in\mathbf{M},p_{i}\in[\hat{a},\infty),\mathbf{p}_{-i}\in[0,\infty)^{n-1}}{q_{i}(\hat{a},\mathbf{p}_{-i},\omega)}\Pi_{i\in N}d\sigma_{i}(p_{i}|m_{i})d\pi d\mu\\
 & -\int_{\omega\in[0,\overline{\omega}],\mathbf{m}\in\mathbf{M},p_{i}\in[\hat{a},\infty),\mathbf{p}_{-i}\in[0,\infty)^{n-1}}{q_{i}(p_{i},\mathbf{p}_{-i},\omega)}\Pi_{i\in N}d\sigma_{i}(p_{i}|m_{i})d\pi d\mu
\end{align*}
The first term satisfies 
\begin{align*}
 & \int_{\omega\in[0,\overline{\omega}],\mathbf{m}\in\mathbf{M},p_{i}\in[\hat{a},\infty),\mathbf{p}_{-i}\in[0,\infty)^{n-1}}{q_{i}(\hat{a},\mathbf{p}_{-i},\omega)\Pi_{i\in N}d\sigma_{i}(p_{i}|m_{i})d\pi d\mu}\\
= & \int_{\omega\in[\hat{a},\overline{\omega}],\mathbf{m}\in\mathbf{M},p_{i}\in[\hat{a},\infty),\mathbf{p}_{-i}\in[\hat{a},\infty)^{n-1}}{q_{i}(\hat{a},\mathbf{p}_{-i},\omega)\Pi_{i\in N}d\sigma_{i}(p_{i}|m_{i})d\pi d\mu} \ \text{(Property \ref{eq:property-of-q} of q)}\\
= & \int_{\omega\in[\hat{a},\overline{\omega}],\mathbf{m}\in\mathbf{M},\mathbf{p}\in(\hat{a},\infty)^{n-1}}{q_{i}(\hat{a},\mathbf{p}_{-i},\omega)\Pi_{i\in N}d\sigma_{i}(p_{i}|m_{i})d\pi d\mu}\ \text{(Equation \ref{eq:atomless})}\\
= & \hat{a}\int_{\omega\in[\hat{a},\overline{\omega}],\mathbf{m}\in\mathbf{M},\mathbf{p}\in(\hat{a},\infty)^{n-1}}{\Pi_{i\in N}d\sigma_{i}(p_{i}|m_{i})d\pi d\mu}\ \text{(Property \ref{eq:property-of-q} of q)}\\
= & \hat{a}G((\hat{a},\infty))=\hat{a}G([\hat{a},\bar{a}])
\end{align*}
The second term is exactly firm $i$'s profit given that the transaction price $p\in [\hat{a},\bar{a}]$. To see this,
\begin{align*}
 & \sum_{i\in N}{\int_{\omega\in[0,\overline{\omega}],\mathbf{m}\in\mathbf{M},p_{i}\in[\hat{a},\infty),\mathbf{p}_{-i}\in[0,\infty)^{n-1}}{q_{i}(p_{i},\mathbf{p}_{-i},\omega)\Pi_{i\in N}d\sigma_{i}(p_{i}|m_{i})d\pi d\mu}}\\
= & \sum_{i\in N}{\int_{p\in[\hat{a},\infty),\omega\in[p,\overline{\omega}],\mathbf{m}\in\mathbf{M},\mathbf{p}_{-i}\in[p,\infty)^{n-1}}{q_{i}(p,\mathbf{p}_{-i},\omega)d\sigma_{i}(p|m_{i})\Pi_{j\ne i}d\sigma_{j}(p_{j}|m_{j})d\pi d\mu}}\ \text{(Property \ref{eq:property-of-q} of q)}\\
= & \int_{p\in[\hat{a},\infty),\omega\in[p,\overline{\omega}],\mathbf{m}\in\mathbf{M}}{\sum_{i\in N}\int_{(p_{i},\mathbf{p}_{-i})\in\{p\}\times[p,\infty)^{n-1}}{q_{i}(p,\mathbf{p}_{-i},\omega)\Pi_{i\in N}d\sigma_{i}(p_{i}|m_{i})}d\pi d\mu}\\
= & \int_{p\in[\hat{a},\infty),\omega\in[p,\overline{\omega}],\mathbf{m}\in\mathbf{M}}{\sum_{i\in N}\int_{(p_{i},\mathbf{p}_{-i})\in[p,\infty)^{n}\backslash(p,\infty)^{n}}{q_{i}(p,\mathbf{p}_{-i},\omega)\Pi_{i\in N}d\sigma_{i}(p_{i}|m_{i})}d\pi d\mu}\ \text{(Property \ref{eq:property-of-q} of q)}\\
= & \int_{p\in[\hat{a},\infty),\omega\in[p,\overline{\omega}],\mathbf{m}\in\mathbf{M}}{pd\sigma_{\underline{p}}(p|\mathbf{m})d\pi d\mu}=\int_{p\in[\hat{a},\infty)}pdG(p)
\end{align*}
The fact that $\hat{a}>\overline{a}/n$ and $G([\hat{a},\overline{a}])>0$ leads to a strictly profitable deviation.

\end{proof}

\section{Discussion and Future Work}
This note simplifies and generalizes the results in \cite{Jann2015}. They consider correlated equilibrium in Bertrand competition, while our results can be interpreted as Bayesian Correlated Equilibrium (c.f. \citealt{bergemann2016bayes}). 

In the future, we want to study if consumers' WTP range $\Omega \in [0,\infty)$ instead of being bounded, can we design a market segmentation profile to induce Folk Theorem as \cite{Baye1999}.\footnote{The Folk Theorem in \cite{Baye1999} requires that the monopoly price being infinity. }

\bibliographystyle{chicago}
\bibliography{mkt-seg-bertrand}

\end{document}